%% file: main-arXiv.tex
\documentclass[letterpaper, 11pt, notitlepage]{article}

\usepackage[margin=1in, pass]{geometry}
\usepackage{hyperref}
\usepackage{amsmath}
\usepackage{amsthm}
\usepackage{enumerate}
\usepackage{fullpage}
\usepackage{appendix}
\usepackage{multicol}
\usepackage{xcolor}
\usepackage[vlined,ruled,linesnumbered,noresetcount]{algorithm2e}
\usepackage{tikz}
\usetikzlibrary{shapes,snakes}

\newtheorem{claim-subsection}{Claim}[subsection]
\newtheorem{definition}{Definition}
\newtheorem{lemma}{Lemma}
\newtheorem{theorem}{Theorem}

\newcommand{\remove}[1]{}

\newcommand{\oh}[1]{{\color{red} #1}\normalcolor}

\DontPrintSemicolon

\makeatletter
\newcommand{\removelatexerror}{\let\@latex@error\@gobble}
\makeatother

\newcommand{\writefunc}{\mbox{\sc Write}}
\newcommand{\readfunc}{\mbox{\sc Read}}
\newcommand{\writerecover}{\mbox{\sc Write.Recover}}
\newcommand{\readrecover}{\mbox{\sc Read.Recover}}
\newcommand{\casfunc}{\mbox{\sc Cas}}
\newcommand{\casrecover}{\mbox{\sc Cas.Recover}}
\newcommand{\Op}{\mbox{\sc Op}}

\newcommand{\init}{\mbox{$\bot$}}

\newcommand{\CAS}{\mbox{\textit{CAS}}}
\newcommand{\True}{\mbox{\textbf{true}}}
\newcommand{\False}{\mbox{\textbf{false}}}
\newcommand{\Ack}{\mbox{\textbf{ack}}}
\newcommand{\fail}{\mbox{\textbf{fail}}}

\title{Upper and Lower Bounds on the Space Complexity \\
	of Detectable Objects}

\begin{document}

	\thispagestyle{empty}
	\begin{titlepage}
		
		\author{
			Ohad Ben-Baruch\\
			Ben-Gurion University, Israel\\
			ohadben@post.bgu.ac.il
			\and
			Danny Hendler\\
			Ben-Gurion University, Israel\\
			hendlerd@cs.bgu.ac.il
			\and
			Matan Rusanovsky\\
			Ben-Gurion University and\\ Israel Atomic Energy Commission, Israel\\
			matanru@post.bgu.ac.il
		}
		
		\maketitle
	
	

\input{abstract}

\end{titlepage}

\input{introduction}

\input{model}

\input{register-upper-bound}

\input{cas-object}

\input{space-lower-bound}

\input{discussion}

	\bibliographystyle{acm}
	\bibliography{references}

	\clearpage
	\appendix

\input{appendix-doubly-perturbing}

\end{document}

%% file: abstract.tex
\begin{abstract}
The emergence of systems with non-volatile main memory (NVM) increases the interest in the design of \emph{recoverable concurrent objects} that are robust to crash-failures, since their operations are able to recover from such failures by using state retained in NVM. Of particular interest are recoverable algorithms that, in addition to ensuring object consistency, also provide \emph{detectability}, a correctness condition requiring that the recovery code can infer if the failed operation was linearized or not and, in the former case, obtain its response.

In this work, we investigate the space complexity of detectable algorithms and the external support they require. We make the following three contributions. First, we present the first wait-free bounded-space detectable read/write and CAS object implementations. Second, we prove that the bit complexity of every $N$-process obstruction-free detectable CAS implementation, assuming values from a domain of size at least $N$, is $\Omega(N)$. Finally, we prove that the following holds for obstruction-free detectable implementations of a large class of objects: their recoverable operations must be provided with \emph{auxiliary state} -- state that is not required by the non-recoverable counterpart implementation -- whose value must be provided from outside the operation, either by the system or by the caller of the operation. In contrast, this external support is, in general, not required if the recoverable algorithm is not detectable.

\end{abstract}

%% file: introduction.tex
\section{Introduction}
\label{section:introduction}

Byte-addressable non-volatile main memory (NVM)
combines the performance benefits of conventional (volatile) main memory
with the durability of secondary storage.
Systems where non-volatile memory co-exists with volatile main memory already exist and are expected to become more prevalent in the future. This increases the interest in the \emph{crash-recovery} model, where a failed process may be resurrected by the system following a crash. Traditional log-based recovery techniques can be applied correctly in such systems but fail to take full advantage of the parallelism and efficiency that may be gained by allowing processing cores to concurrently access recovery data directly from NVM, rather than by performing slow block transfers from secondary storage. Consequently, there is increasing interest in the design of \emph{recoverable concurrent objects} that are robust
to crash-failures, since their operations are able to recover from such failures by using state retained in NVM (see e.g. \cite{AttiyaBH18,FriedmanHMP18,GolabH17recoverable,GolabH18recoverable,GolabR16recoverable,JayantiJJ18,JayantiJ17}).

Of particular interest are recoverable algorithms that, in addition to ensuring object consistency, also provide \emph{detectability} \cite{FriedmanHMP18}. Detectability requires that the code recovering from a failed operation can infer if it was linearized or not and, in the former case, obtain its response.
Several recent works presented detectable algorithms \cite{AttiyaBH18,Ben-DavidBFW19,FriedmanHMP18}. In particular, both Ben-David et al.~\cite{Ben-DavidBFW19} and Attiya et al. \cite{AttiyaBH18} presented detectable CAS algorithms and~\cite{AttiyaBH18} also presented a detectable read/write object. All these algorithms augment the arguments of recoverable operations with unique identifiers, for allowing the recovery code to detect whether or not the failed operation was linearized, consequently incurring unbounded space complexity.\footnote{This is also the case with the durable queue algorithm of \cite{FriedmanHMP18}.} In addition, \cite{AttiyaBH18} proved that every lock-free detectable test-and-set implementation from (non-recoverable) test-and-set objects must use unbounded space. This raises the question of whether unbounded space complexity is inherent to nonblocking detectable implementations of these objects. \emph{We provide a negative answer to this question by presenting the first nonblocking bounded-space detectable CAS and read/write algorithms}. Both algorithms are wait-free. Our $N$-process bounded-space CAS algorithm uses $\Theta(N)$ bits in addition to those storing the CAS object's value. In our second contribution, \emph{we show that every obstruction-free detectable CAS implementation, assuming values from a domain of size at least $N$, must have $\Omega(2^N)$ different reachable shared-memory configurations}, thus establishing that our CAS algorithm's space complexity is asymptotically optimal.

Detectable algorithms often require \emph{auxiliary state} that helps them infer where in the execution the failure occurred. Informally, auxiliary state is information that is provided to the recoverable operation that is not provided to (nor required by) the ``original'' (non-recoverable) operation. In some works, it is assumed that this information is provided by the system.
For example, the recoverable mutual exclusion algorithms presented by Golab and Hendler \cite{GolabH18recoverable} assume a model in which the system provides to each operation an \emph{epoch number} whose value increases after each (system-wide) failure. Some detectable algorithms presented by Attiya et al. \cite{AttiyaBH18} assume that the system provides to the recovery code information identifying the instruction that the failed operation was about to execute via a non-volatile variable.
However, auxiliary state is not necessarily provided by the system.
For example, the read/write algorithm of \cite{AttiyaBH18}, the CAS algorithm of Ben-David et al. \cite{Ben-DavidBFW19} and the queue algorithm of Friedman et al. \cite{FriedmanHMP18} rely on auxiliary state (e.g. unique identifiers) passed to recoverable operations via their arguments by the operations that invoke them.


We show that, for a large class of objects that includes read/write, CAS and FIFO queue objects, any obstruction-free detectable implementation must receive auxiliary state. As we prove, this auxiliary state must be made available to recoverable operations either via their arguments or via a non-volatile variable accessible by them whose value must be modified outside the operation. In contrast, this external support is, in general, not required if the recoverable algorithm is not detectable.

The rest of the paper is organized as follows. We describe the system model in Section \ref{section: Model}. We then present our bounded-space detectable read/write and CAS algorithms in Sections \ref{section: read-write} and \ref{section: CAS}, respectively. In Section \ref{section: CAS} we also prove a lower bound on the space complexity of detectable CAS. This is followed by a proof that detectable implementations of a large class of objects require auxiliary state in Section \ref{section: auxiliary-state}. The paper is concluded by a short discussion in Section \ref{section: discussion}.

%% file: model.tex
\section{System Model}
\label{section: Model}

A set $P$ of $N$ asynchronous crash-prone \emph{processes}
communicate through \emph{shared objects}.
The system provides \emph{base objects} (also called shared variables or registers) that support atomic read, write,
and read-modify-write \emph{primitive operations}.
Base objects are used to \emph{implement} higher-level concurrent
objects by defining algorithms, for each process, which use primitive operations to carry
out the operations of the implemented object.

The state of the system consists of \emph{non-volatile shared-memory variables} and per-process \emph{local variables} stored in its local \emph{volatile} cache. Local variables are accessed only by the process to which they belong. For presentation simplicity, we assume that each process $p$ may own \emph{non-volatile private variables} that reside in the NVM but are accessed only by $p$.
We also assume the abstract \emph{private cache model}
\cite{Ben-DavidBFW19,IzraelevitzMS16}, in which all shared variables are always persistent and there is no shared cache. In this model, primitive operations to shared variables are applied directly to the NVM. At any point during the execution of an operation,
a \emph{system-wide crash-failure}
(or simply a \emph{crash}) may occur, which resets the local variables of all processes to their initial values, but preserves the values of all non-volatile variables.

As we explain in Section \ref{section: discussion}, all our results hold also in the more realistic \emph{shared-cache model}. In this model, in addition to per-process private caches, there is a single (volatile) shared cache. Primitive operations to shared variables are applied to this cache and explicit persistency instructions may be required for guaranteeing that values written to this cache get persisted to the NVM in the correct order \cite{IzraelevitzMS16}.

To start executing an operation $Op$, a process $p$ \emph{invokes} $Op$. We say that \emph{Op} \emph{completes} once control returns to the caller of $Op$. Before completing, $Op$ returns a \emph{response value}, which is stored to a local variable of $p$.
The response value is lost if $p$ crashes before \emph{persisting} it
(i.e., writing it to a non-volatile variable). We say that a process is \emph{idle} if it is not in the midst of executing any operation. Each recoverable operation $Op$ of a shared object
is associated with a \emph{recovery function}, denoted $Op.\texttt{Recover}$, which is responsible to infer whether $Op$ was linearized or not, and to obtain its response in the former case. $Op.\texttt{Recover}$ is performed by $p$ in order to recover from a failure that occurred while $p$ was executing $Op$. We assume that $Op.\texttt{Recover}$ is being called with the same arguments as those with
which $Op$ was invoked when the crash occurred. If $Op.\texttt{Recover}$ infers that $Op$ was not yet linearized, it returns a special \fail\ value, otherwise it returns $Op$'s response.

\emph{Our lower bounds (Theorems \ref{theorem:cas-lower-bound} and \ref{theorem:auxiliary-data-required}) only require the model assumptions specified above.} However, as we prove in Theorem \ref{theorem:auxiliary-data-required}, detectable algorithms must receive auxiliary state whose value is modified either by the operation's caller or by the system. We therefore make the following additional assumptions that are used by the algorithms we present. Each process $p$ is associated with a private non-volatile structure $Ann_p$ consisting of three fields. $Ann_p.op$ stores the type of recoverable operation currently performed by $p$, as well as the arguments with which it was called. It is accessed only by the caller of the recoverable operation $Op$, which sets its value (thus announcing the operation it is about to perform) immediately before invoking $Op$.
Which function (if any) should be invoked by $p$ in order to recover from a failure is determined according to the value of $Ann_p.op$. Field $Ann_p.resp$ stores the response of the recoverable operation and is initialized to $\init$ immediately before $Op$ is invoked. The 3rd field, $Ann_p.CP$, may be used by recoverable operations and recovery functions for managing checkpoints in their execution flow. Field $Ann_p.CP$ is set to 0 by the caller of the recoverable operation immediately before invoking it. $Ann_p.CP$ can be read and written by recoverable operations and their recovery functions and
is used by $p$ in order to record (in the NVM) the fact that the execution reached a certain point. The recovery function can then use this information in order to correctly recover and to avoid re-execution of critical instructions.

Failed processes recover in an asynchronous manner, independently of each other. Specifically, the recovery of some processes may have already completed while other processes may have not yet completed (or even started) their recovery. $Op.\texttt{Recover}$ may be invoked multiple times
before it completes, because the system may undergo multiple crashes in the course of executing it.
If all the operations of an implementation are recoverable,
then the implementation is called {\em recoverable}.


Linearizability \cite{HerlihyW90} requires that each operation applied to a concurrent object takes effect instantaneously at some point between its invocation and response. The correctness condition ensured by our algorithms is \emph{durable linearizability} (DL) \cite{IzraelevitzMS16}. DL requires that linearizability be maintained in spite of crash-failures. In other words, once the system recovers after a crash-failure, the state of the data structure reflects a history containing all operations that completed before the crash and may also contain some operations that have not completed before the crash. This  captures the idea that an operation can be linearized only once its effect gets persisted to NVM.


The progress conditions we consider are \emph{wait-freedom} \cite{Herlihy91} and \emph{obstruction-freedom} \cite{HerlihyLM03}. A recoverable operation or a recovery function is wait-free (resp. obstruction-free) if, starting from any reachable configuration, $p$ completes it in a finite number of its own steps (resp. when running solo), when the system experiences no crashes.
We emphasize that all our results hold also in a model where processes may fail independently, such as that assumed by \cite{AttiyaBH18}.

%% file: register-upper-bound.tex
\section{Detectable Read/Write Object}
\label{section: read-write}
Algorithm \ref{bounded-space-recoverable-register} presents the pseudo-code of a detectable read/write object $O$ that uses bounded space from (bounded-space) variables that support read/write primitive operations.
To the best of our knowledge, this is the first detectable read/write algorithm that uses bounded space.
The checkpoint field $Ann_p.CP$ is used by process $p$ in order to allow the recovery function to infer where in the recoverable operation the failure occurred.
Each process $p$ owns two private variables: $RD_p$, storing data used during recovery, and $T_p$, storing an index $\in \{0,1\}$ to one of two size-$n$ \emph{toggle-bit arrays}, $A[\ ][p][0]$ and $A[\ ][p][1]$, that are used by $p$'s write operations in an alternating manner. $O$'s state consists of a single shared read/write register $R$ storing a triplet of values $\langle v,q,b \rangle$, where $v$ is $O$'s current value, $q$ is the identifier of the process that (last) wrote $v$, and $b$ is the index of the toggle-bit array used by $q$ for that write operation. Initially, $R=\langle v_{init},0,0 \rangle$, where $v_{init}$ is $O$'s initial value, thus ``attributing'' this value to a write by process $0$ that used toggle-bit array $0$. Register $R$ stores $O(\log n)$ bits in addition to the application value $v$, in contrast with the unbounded state required by the read/write object implementation of Attiya et al. \cite{AttiyaBH18}. A 3-dimensional array $A[N][N][2]$ allows each writing process $p$ to coordinate with any other process $q$ using $p$'s two toggle-bit arrays.

The key challenge with which Algorithm \ref{bounded-space-recoverable-register} copes is the ABA problem. Attiya et al. \cite{AttiyaBH18} avoid it by ensuring that all written values are distinct, at the cost of using a register of unbounded-size. Algorithm \ref{bounded-space-recoverable-register} allows the same value to be written multiple times, so a process $p$ may read from $R$ a value $v_q$ (written by process $q$) and then write some value $v_p$ that is later overwritten by another write of $v_q$ by $q$. In this case, if $p$ recovers after a system crash, a mechanism for allowing it to detect whether or not its operation was linearized is required. As we explain below, per-process toggle bits are used to implement this mechanism. Before invoking an operation on the object, its caller initializes the $Ann_p$ structure as described in Section \ref{section: Model}. Specifically, $Ann_p.CP$ is initialized to $0$ and $Ann_p.resp$ is initialized to $\init$.

\paragraph{The \writefunc\ operation} To write, process $p$ reads $R$ (line \ref{write-enter}), thus learning that $q$ was the last to write to $R$ and which toggle-bit array was used by $q$ for writing. Next, $p$ resets the bit from $q$'s \emph{other} toggle-bit array corresponding to $p$ (line \ref{write-set-0}), and persists the value read from $R$, as well as the index of the toggle-bit array used by $p$'s current write (stored in $T_p$), into $RD_p$ (lines \ref{write-read-Tp}-\ref{write-RD-update}). Then, $p$ reads $R$ again (line \ref{write-read-R-again}) and proceeds to write to $R$ (line \ref{write-R-update}) only if it read from $R$ the same value as in line \ref{write-enter}. In this case, $p$ sets its checkpoint field to $1$ (line \ref{write-CP-first}) immediately before the write to $R$ and sets it to $2$ (line \ref{write-CP-termination}) immediately after it. It then sets all the bits in the toggle-bit array used by its current write operation, switches its toggle-bit array index, persist the response and returns (lines \ref{write-for}-\ref{write-return}).

If the condition of line \ref{write-read-R-again} is not satisfied then, as we prove, a write operation $W$ by a process other than $p$ is linearized between $p$'s first and second reads of $R$, hence $p$ can be assumed to have been overwritten by $W$. In this case, $p$ skips lines \ref{write-CP-first}-\ref{write-R-update} and proceeds directly to line \ref{write-CP-termination}. 

\paragraph{The \writerecover\ recovery function}

Upon recovery from a failed \writefunc\ operation $W$, $p$ first reads $RD_p$ (line \ref{write-recover-read-RD}) and then checks if $Ann_p.result$ was set (line \ref{write-recover-if-result-set}). In this case, $W$ was completed and has been linearized, so the recovery function returns \Ack. Next, $p$ checks if $Ann_p.CP$ equals $0$ (line \ref{write-recover-if-no-CP}). In this case, as we prove, $W$ was not linearized before the failure, so the recovery function returns \fail\ (line \ref{write-reinvoke-if-CP-0}); the caller of the failed operation can now decide whether or not to reattempt performing $W$. Otherwise, if $Ann_p.CP$ equals $1$ (line \ref{write-recover-if-after-first-CP}), then the recovery code must determine whether or not $R$ was written in line \ref{write-R-update} (either by $p$ or by another process) since when $W$ read $\langle qval, q,qtoggle \rangle$ from $R$ in line \ref{write-enter}. This is done in line \ref{write-recover-if-no-write} as follows. If $R$'s value differs from $\langle qval, q,qtoggle \rangle$, then $R$ was written and so either $W$ performed line \ref{write-R-update} or $W$ can be assumed to have been overwritten by another write, so the recovery code proceeds by performing lines \ref{write-recover-set-CP}-\ref{write-recover-return} (which are identical to lines \ref{write-CP-termination}-\ref{write-return}). Otherwise, $R$'s value equals $\langle qval, q,qtoggle \rangle$ but it is still possible that $q$ wrote $\langle qval, q,qtoggle \rangle$ to $R$ again after $R$ was read by $W$. This is checked by the 2nd condition of line \ref{write-recover-if-no-write} which relies
on the following key observation used by our correctness proof: in order for $q$ to write again using the same toggle-bit index, it must first complete a write operation using the other toggle-bit index. However, in that earlier write operation, $q$ sets all its toggle bits of that set to 1 (either in lines \ref{write-for}-\ref{write-for-end} of its write operation or in lines \ref{write-recover-for}-\ref{write-recover-set-toggle-bits} of its recovery function). Therefore, upon recovery, if $p$ reads the same value from $R$ as before the crash, it can conclude that a write occurred in between its two reads of $R$ if and only if $q$'s toggle bit that it has set to 0 is now 1. If this is not the case, $p$ concludes that $W$ was not linearized and returns \fail\ (line \ref{write-recover-no-write-fail}).

The \readfunc\ operation reads a triplet of values from $R$ and then extracts its first component, writes it to $Ann_p.resp$ and returns it. Its recovery function re-invokes \readfunc\ if $Ann_p.resp=\init$ holds, otherwise it returns it. This simple code is not presented in Algorithm \ref{bounded-space-recoverable-register}.

It is easily seen that Algorithm \ref{bounded-space-recoverable-register} uses bounded space, assuming that the values written by \writefunc\ operations are of bounded size. It remains to show that the algorithm satisfies durable linearizability, detectability and wait-freedom.

\begin{lemma}
Algorithm \ref{bounded-space-recoverable-register} is wait-free and satisfies durable linearizability and detectability.
\end{lemma}

\begin{proof}
Consider an execution $\alpha$ of Algorithm \ref{bounded-space-recoverable-register}. Assume process $p$ completes  a \writefunc ($val$) operation $W$ in $\alpha$ (either directly or by completing the recovery function). We prove that one of the following holds: 1) $p$ writes to $R$ exactly once, and this is $W$'s linearization point; 2) $p$ does not write to $R$ and there is a concurrent write operation $W'$ by a different process that writes to $R$, hence we can linearize $W$ immediately before $W'$; or 3) the failure occurred before $W$ wrote to $R$, in which case \writerecover\ returns \fail.

\begin{figure*}[!t]
	\removelatexerror
	
	\begin{algorithm}[H]
		
		\footnotesize
		\begin{flushleft}
			\textbf{Non-Volatile Shared variables}: read/write register $R$ initially $\langle v_{init},0,0 \rangle$, boolean $A[N][N][2]$ initially all $0$ \\
			\textbf{Non-Volatile Private variables}: Read/write register $RD_p$ initially $\bot$, $T_p$ initially 0\\

		\end{flushleft}
		
		\begin{multicols*}{2}
			
			\begin{procedure}[H]
				\caption{() \small \writefunc\ (val)}
				
				$\langle qval, q,qtoggle \rangle := R$ \label{write-enter} \;
				$A[p][q][1-qtoggle] := 0$ \label{write-set-0} \;
                $mtoggle := T_p$ \label{write-read-Tp} \;
				$RD_p := \langle mtoggle,qval, q,qtoggle \rangle$ \label{write-RD-update} \;
				\lIf {$R \neq \langle qval, q,qtoggle \rangle$ \label{write-read-R-again}} {
					\textbf{goto} \ref{write-CP-termination} \label{write-goto}
				}
				$Ann_p.CP := 1$ \label{write-CP-first} \;
				$R := \langle val, p, mtoggle \rangle$ \label{write-R-update} \;
				$Ann_p.CP := 2$ \label{write-CP-termination} \;
				\For {$i=1$ to $N$ \label{write-for}} {
					$A[i][p][mtoggle] := 1$
				}\label{write-for-end}
                $T_p := 1-mtoggle$\;
                $Ann_p.result := \Ack$ \;
				\KwRet \Ack \label{write-return}
			\end{procedure}
			
			\remove{
			\begin{procedure}[H]
				\caption{() \small  \readfunc\ ()}
				
				$\langle val,p,toggle \rangle := R$ \label{BSreg-read:read-R}\;
				\KwRet val \label{BSreg-read:return-val}
			\end{procedure}
			}
			
			\columnbreak
			
			\begin{procedure}[H]
				\caption{() \small \writerecover\ (val)}
				
				$\langle mtoggle, qval, q,qtoggle \rangle := RD_p$ \label{write-recover-read-RD}\;
				\uIf {$Ann_p.result \neq \init$ \label{write-recover-if-result-set}} {
					\KwRet \Ack
				}
				\uIf {$Ann_p.CP = 0$ \label{write-recover-if-no-CP}} {
					\KwRet \fail \label{write-reinvoke-if-CP-0}
				}
				\uIf {$Ann_p.CP = 1$ \label{write-recover-if-after-first-CP}} {
					\uIf {$R = \langle qval, q,qtoggle \rangle$ and $A[p][q][1-qtoggle] = 0$ \label{write-recover-if-no-write}} {
						\KwRet \fail \label{write-recover-no-write-fail}
					}
				}
				$Ann_p.CP := 2$ \label{write-recover-set-CP} \;
				\For {$i=1$ to $N$ \label{write-recover-for}} {
					$A[i][p][mtoggle] := 1$ \label{write-recover-set-toggle-bits}\;
				} \label{write-recover-for-end}
                $T_p := 1-mtoggle$\;
                $Ann_p.result := \Ack$ \;
				\KwRet \Ack \label{write-recover-return}
			\end{procedure}
			
			\remove{
			\begin{procedure}[H]
				\caption{() \small  \readrecover\ ()}
				
				Re-invoke \readfunc
			\end{procedure}
			}
			
		\end{multicols*}
		
		\caption{Bounded space detectable read/write object $O$, code for process $p$.}
		\label{bounded-space-recoverable-register}
	\end{algorithm}
	
\end{figure*}

We start by observing that if there is no crash during $W$, then either it writes to $R$ exactly once in line \ref{write-R-update}, or the condition in line \ref{write-goto} holds, in which case $p$ does not write to $R$ and there was a concurrent write operation $W'$ by a different process that wrote to $R$.

We proceed to prove that the lemma holds also if $W$ and its recovery code experience one or more crashes. If the system crashes before $p$ writes to $\text{CP}_p$ in line \ref{write-CP-first} 
then $p$ did not write to $R$ while executing $W$, nor did it write to any entry in its toggle-bit arrays $A[i][p][y]$ for $i \neq p$, hence none of its writes is ever read by another process. Thus, $W$ was not linearized yet and the recovery code simply returns \fail\ (line \ref{write-reinvoke-if-CP-0}).
Otherwise, consider the case where a crash occurred after $p$ executed line \ref{write-CP-termination}. Then either $p$  wrote to $R$ in line \ref{write-R-update} (hence $W$ was linearized) or $p$ observed in line \ref{write-goto} a write by a concurrent write operation $W'$ and so $W$ can be linearized immediately before $W'$. In both cases, \writerecover\ re-executes the for loop of lines \ref{write-recover-for}-\ref{write-recover-for-end}, switches its toggle-bit index, and returns.

We are left with the case where a crash occurred after the update of $\text{CP}_p$ in line \ref{write-CP-first} but before its update in line \ref{write-CP-termination}. Upon recovery, in order to satisfy detectability, the algorithm has to determine whether $W$ was linearized or not.
To complete the proof, we establish the following two claims: 1) the condition of line \ref{write-recover-if-no-write} evaluates to false only if there was a write to $R$ (either by $p$ or by some other process) after $p$ first read it (in line \ref{write-enter}); 2) if $p$ wrote to $R$ before the crash (in line \ref{write-R-update}), then the condition of line \ref{write-recover-if-no-write} evaluates to false.

In line \ref{write-enter}, $p$ reads $R$ and writes its content (together with $p$'s current value of $T_p$) to $RD_p$ in line \ref{write-RD-update}. Therefore, $p$ persists the identifer $q$ of the process that last wrote to $R$ and the toggle-bit index used by $q$. Denote by $W_q$ this write of $q$, and let \emph{qts} denote the toggle-bit index used in $W_q$. Assume that the condition of line \ref{write-recover-if-no-write} evaluates to false. We prove that some write operation is linearized after $W_q$. Since a write operation can only be linearized when $R$ is written (not necessarily by the linearized operation itself), this would prove claim 1). If $p$ reads from $R$ in line \ref{write-recover-if-no-write} a value different from the one stored in $RD_p$, then clearly there was a write to $R$ after the read in line \ref{write-enter} and we are done. Otherwise, $p$ reads the same value from $R$ in line \ref{write-recover-if-no-write}, and so it holds that $A[p][q][1-qts] = 1$. In line \ref{write-set-0}, $p$ sets $A[p][q][1-qts]$ to 0. Thus, $q$ must have set it to 1 later in the execution. Notice that $q$ sets a bit of the ($1-qts$) toggle-bit array to 1 only during the for loop (in lines \ref{write-for}-\ref{write-for-end} or \ref{write-recover-for}-\ref{write-recover-set-toggle-bits}), after completing a write using the toggle-bit index $1-qts$. Denote this write by $\overline{W}_q$. Since $p$ observed the write $W_q$, which is associated with the toggle-bit index $qts$, it must be that write $\overline{W}_q$ was performed after the read of $p$ in line \ref{write-enter}. This is true since $\overline{W}_q$ must either write to $R$, or observe a write to $R$ by another process, and in both cases, the linearization point of $\overline{W}_q$ must be after $p$ observed $W_q$ in line \ref{write-enter}.

We now prove claim 2). Assume the system crashed after $p$ wrote to $R$ in line \ref{write-R-update}. Upon recovery, if $p$ reads in line \ref{write-recover-if-no-write} a value from $R$ other than that stored in $RD_p$, the claim clearly holds. Otherwise, $p$ reads the same value written by $W_q$. Notice that $p$ reads the same value in line \ref{write-goto}, after setting $A[p][q][1-qts]$ to 0. Moreover, later in the execution $p$ writes to $R$, and thus there must be another write by $q$ to $R$ with the same value as in $W_q$ that was done after $p$ wrote to $R$. In particular, these two writes use the same toggle-bit index \emph{qts}, thus there must be another write by $q$ using toggle-bit index $1-qts$ in between the two. This write must have been completed, and thus the for loop updating $A[p][q][1-qts]$ to 1 was done after $p$'s read in line \ref{write-goto}, hence, after $p$ has set $A[p][q][1-qts]$ to 0. Moreover, no process but $p$ can set $A[p][q][1-qts]$ to 0 again. Consequently, the second condition of the if statement in line \ref{write-recover-if-no-write} does not hold, implying that claim 2) holds.

To conclude the proof of durable linearizability and detectability, we claim that every completed \readfunc\ operation returns the value of the latest \writefunc\ operation linearized before it, or $v_{init}$ if there are no such \writefunc\ operations. We linearize a \readfunc\ operation $Op$ when it writes to $Ann_p.resp$ and its linearization point is then set to its previous read of $R$. It follows that the response of a completed \readfunc\ is the value of $R$'s first component when last read by the operation. A \writefunc\ operation $W$ can be linearized either when it writes to $R$ (line \ref{write-R-update}) or when it does not write to $R$ but is linearized before a concurrent \writefunc\ operation $W'$ that does write to $R$. The claim now follows from the atomicity of $R$.

It is easily seen that the algorithm is wait-free, since neither \readfunc, \writefunc\ or their recovery functions contain any loops, so each of these operations/functions terminates in a constant number of steps if it experiences no crashes.
\end{proof}

%% file: cas-object.tex
\section{A Detectable CAS Algorithm and Lower Bound}
\label{section: CAS}

We now present a wait-free detectable implementation of an $N$-process durable linearizable CAS object from (bounded-space) variables that support read/CAS primitive operations. As far as we know, this is the first bounded-space detectable CAS implementation. Our algorithm uses $\Theta(N)$ bits in addition to those storing the CAS object's value. We then prove that any such implementation requires $\Omega(N)$ bits, thus establishing that our algorithm is asymptotically space optimal.

\subsection{A Bounded-Space Recoverable CAS Algorithm}

Our implementation of a detectable CAS object $O$ is presented by Algorithm \ref{alg:cas}. $O$'s state is represented by a shared variable $C$ (supporting primitive read/CAS operations) that stores both $O$'s value (initially $v_{init}$) and an $N$-bit vector $vec$ (all of whose components are initially $0$). Each process $p$ owns a private variable $RD_p$, storing $p$'s recovery data.

\paragraph{The \casfunc\ operation}
To perform a \casfunc\ operation, $p$ first reads $C$. If $O$'s current value differs from argument \emph{old}, the CAS should fail, so $p$ persists a \emph{false} response to $Ann_p.result$ and then returns \emph{false} (lines \ref{CASwrite-readC}-\ref{CASwrite-returnFalse}).
Otherwise, $p$ flips the bit $vec[p]$ (line \ref{CASwrite-flipBit}), persists the flipped value, and increments its checkpoint variable (lines \ref{CASwrite-persistToggle}-\ref{CASwrite-setCP1}).
Then, $p$ attempts to atomically both change $C$'s value and flip $vec[p]$ using an atomic \CAS\ (line \ref{CASwrite-executeCAS}). Finally, it persists the \CAS\ response and returns it (lines \ref{CASwrite-persist-response}-\ref{CASwrite-return}).

\paragraph{The \casrecover\ recovery function}
Upon recovery from a crash inside \casfunc, $p$ first checks if it crashed after persisting the response, in which case it returns it (lines \ref{CASwrite-recover-if-already-persisted}-\ref{CASwrite-recover-return-result}). Otherwise, if $Ann_p.CP$ is 0, implying $p$ surely crashed before attempting the \CAS, then it returns \emph{fail} (lines \ref{CASwrite-recover-if-cp-0}-\ref{CASwrite-recover-reinvoke}). Finally, $p$ must determine if it performed a successful \CAS\ before crashing. Observe that $p$ is the only process to ever change $vec[p]$, hence a successful \CAS\ operation flips $vec[p]$ and it will remain flipped until $p$'s next successful \CAS, whereas a failed \CAS\ attempt does not change $vec[p]$. Leveraging this observation, $p$ reads the vector stored in $C$ and returns true if $vec[p]$ has been changed (lines \ref{CASwrite-recover-readC}-\ref{CASwrite-recover-if-not-flipped-bit}, \ref{CASwrite-recover-set-result-true}-\ref{CASwrite-recover-return-true}). If the bit was not changed (implying that either the \CAS\ failed or the crash occurred before $p$ performed it), $p$ returns \emph{fail} (line \ref{CASwrite-recover-reinvoke1}).

A \readfunc\ operation simply reads $C$, extracts $O$'s value, writes it to $Ann_p.resp$ and returns it. Its recovery function re-invokes \readfunc\ if $Ann_p.resp=\init$ holds, otherwise it returns it. This simple code is not shown.

\begin{lemma}
	Algorithm \ref{alg:cas} is wait-free and satisfies durable linearizability and detectability.
\end{lemma}

\begin{proof}
	Consider an execution $\alpha$ of Algorithm \ref{alg:cas}. Assume process $p$ completes a \casfunc ($old, val$) operation $Op$ in $\alpha$ (either directly or by completing the recovery function). We prove that one of the following holds:
	1) $p$ successfully writes to $C$ exactly once and the return value of $Op$ is true, in which case this is $Op$'s linearization point;
	2) $p$ does not write to $C$ and $C$ contains a value different then $old$ at some point during $Op$, hence we can linearize $Op$ at this point and it returns false;
	or 3) a failure occurs before $Op$ wrote to $C$, in which case \casrecover\ returns \fail.

	We start by observing that if $Op$ does not experience a crash, then either $p$ reads a value different from $old$ in $C$, thus it returns false (lines \ref{CASwrite-val-diff-old}-\ref{CASwrite-returnFalse}), or $p$ performs a single \CAS\ in line \ref{CASwrite-executeCAS} and returns its response. The \CAS\ is successful only if the value stored in $C$ is $old$, in which case $p$ returns true. On the other hand, the \CAS\ fails only if another process performed a successful \CAS\ to $C$ after $p$ first read it in line \ref{CASwrite-readC}, hence the value of $C$ after it must be other than $old$, and this is also the linearization point of $Op$, which indeed returns false.

	Note that $p$ is the only process to ever update $vec[p]$ and the only place in which this update occurs is the \CAS\ of line \ref{CASwrite-executeCAS}. Moreover, this is the only place in the code where a \CAS\ is performed. Thus, since lines \ref{CASwrite-flipBit}-\ref{CASwrite-persistToggle} have to be executed before line \ref{CASwrite-executeCAS} (even in case of a crash), each successful \CAS\ to $C$ by $p$ will flip the bit $vec[p]$ stored in $C$, and it will remain flipped until $p$'s next successful \CAS.
	
	We proceed to prove that the lemma holds also if $Op$ experiences one or more crashes.
	If a crash occurs before $Op$ writes to $Ann_p.CP$ in line \ref{CASwrite-setCP1},
	then $p$ did not write to $C$ while executing $Op$. Thus, $Op$ was not linearized yet and \casrecover\ simply re-invokes $Op$ (line \ref{CASwrite-recover-reinvoke}).
	Otherwise, consider the case where a crash occurs after $p$ performs line \ref{CASwrite-setCP1}. If $p$ performed a successful \CAS\ at line \ref{CASwrite-executeCAS} before the crash, then the operation was already linearized. As per our previous observation, $C.vec[p] = RD_p$ will hold as long as $p$ does not perform another (successful) \CAS. Hence, the condition in line \ref{CASwrite-recover-if-not-flipped-bit} is evaluated as false, and $p$ persists $\True$ as its response and then returns (lines \ref{CASwrite-recover-set-result-true}-\ref{CASwrite-recover-return-true}).

	It remains to consider the case where the crash occurred when $p$ did not perform a successful \CAS\ at line \ref{CASwrite-executeCAS}, either because the \CAS\ failed or because $p$ crashed before performing line \ref{CASwrite-executeCAS}. In this case, we can consider $Op$ as not having been linearized yet, since it did not change the value of any variable that operations by other processes may read.
	In both cases, $vec[p] \neq RD_p$ holds, since $vec[p]$ stores the old, un-flipped, value, whereas $RD_p$ stores the new, flipped, value. Moreover, no process but $p$ can change $vec[p]$. Thus, the condition of line \ref{CASwrite-recover-if-not-flipped-bit} is evaluated as true and the recovery function returns \fail\ in line \ref{CASwrite-recover-reinvoke1}.
\end{proof}

\begin{figure*}[!t]
	\removelatexerror
	
	\begin{algorithm}[H]
		
		\footnotesize
		
		\begin{flushleft}
			\textbf{Non-Volatile Shared variables}: $C$, supporting primitive read/CAS operations, initially $\langle v_{init}, \langle 0,\ldots,0 \rangle \rangle$ \\
			\textbf{Non-Volatile Private variables}: read/write register $RD_p$ containing boolean field, initially $\bot$
		\end{flushleft}
		
		\begin{multicols*}{2}

			\begin{procedure}[H]
				\caption{() \small \casfunc\ (old, new)}
				
				$\langle val, vec \rangle := C$ \label{CASwrite-readC}\;
				\uIf (\tcp*[f]{CAS failed; return}) {$val \neq old$ \label{CASwrite-val-diff-old}} {
					$Ann_p.result := \False$ \;
					\KwRet \False \label{CASwrite-returnFalse}
				}
				$newvec := flipBit(vec,p)$ \tcp*{flip \emph{vec}[p]} \label{CASwrite-flipBit}
				$RD_p := newvec[p]$ \tcp*{persist new bit} \label{CASwrite-persistToggle}
				$Ann_p.CP := 1$ \label{cas-CP} \tcp*{set check-point} \label{CASwrite-setCP1}
				$res := C.\CAS (\langle val, vec \rangle, \langle new, newvec \rangle)$ \label{CASwrite-executeCAS}\;
				$Ann_p.result := res$ \label{CASwrite-persist-response} \tcp*{persist response}
				\KwRet $res$ \label{CASwrite-return}  \;
			\end{procedure}
			
			\columnbreak
			
			\begin{procedure}[H]
				\caption{() \small \casrecover\ (old, new)}
				
				\uIf {$Ann_p.result \neq \init$ \label{CASwrite-recover-if-already-persisted}} {
					\KwRet $Ann_p.result$ \label{CASwrite-recover-return-result}
				}
				\uIf {$Ann_p.CP = 0$ \label{CASwrite-recover-if-cp-0}} {
					\KwRet \fail \label{CASwrite-recover-reinvoke}
				}
				$\langle val, vec \rangle := C$ \label{CASwrite-recover-readC}\;
				\uIf {$vec[p] \neq RD_p$ \label{CASwrite-recover-if-not-flipped-bit}} {
					\KwRet \fail \tcp*{CAS failed or not performed} \label{CASwrite-recover-reinvoke1}
				}
				$Ann_p.result := \True$ \tcp*{CAS was successful} \label{CASwrite-recover-set-result-true}
				\KwRet \True \label{CASwrite-recover-return-true} \;
			\end{procedure}

		\end{multicols*}
		\caption{Bounded-space detectable CAS object $O$, code for process $p$.}
		\label{alg:cas}
	\end{algorithm}
	
\end{figure*}

\subsection{A Space Lower Bound on Detectable CAS}

Algorithm \ref{alg:cas} uses $\Theta(N)$ shared memory bits, in addition to those required for storing the CAS object's value.
Let $\mathcal{V}$ denote the set of states that may be assumed by a CAS object $O$. Assuming that $\vert {\mathcal{V}} \vert \geq N$, we prove that any recoverable and detectable obstruction-free implementation $A$ of $O$ must have at least $2^{N-1}$ reachable configurations with distinct shared-memory states. This implies that $A$ uses at least $N-1$ shared-memory bits. If $\vert {\mathcal{V}} \vert = O(N)$, then only $O(\log N)$ bits are required for storing the CAS object's value, implying that  $\Omega(N)$ additional shared-memory bits are required for supporting detectability. Before proceeding with the proof, we need the following two definitions.  We say that configurations $C,D$ are \emph{memory-equivalent} if the values of each shared memory variable is the same in both configurations. We say that a step $s$ by some process $p$ is a \emph{modifying step} \cite{HendlerS08} w.r.t. to an operation $Op$ by another process $q$ in some configuration $C$, if the solo execution of $Op$ by $q$ after $C$ and after $C \circ s$ return different responses. Figure \ref{fig:impos-cas} illustrates the structure of the inductive construction of the proof that follows.

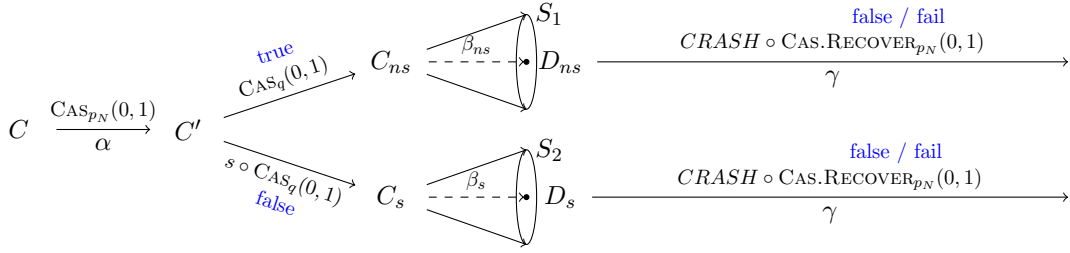
\begin{figure*}[!t]
	\centering
	\begin{tikzpicture}
	\begin{scope}[scale=0.9, transform shape]
	\node[circle, minimum height=1.1cm] (c) at (-1.5,0) {$C$};
	\node[circle, minimum height=1.1cm] (ct) at (1,0) {$C'$};

	\node[circle, minimum height=1.1cm] (cs) at (4,-1) {$C_{s}$};
	\node[circle, minimum height=1.1cm] (cns) at (4,1) {$C_{ns}$};
	\node[circle, scale=0,label=right:$S_1$] (tns) at (6,1.7) {};
	\node[circle, scale=0] (bns) at (6,0.3) {};
	\draw (6,1) ellipse (0.15cm and 0.7cm);
	\node[circle, scale=0.2, draw, fill] (mns) at (6,1) {};
	
	\node[circle, scale=0,label=right:$S_2$] (ts) at (6,-0.3) {};
	\node[circle, scale=0] (bs) at (6,-1.7) {};
	\draw (6,-1) ellipse (0.15cm and 0.7cm);
	\node[circle, scale=0.2, draw, fill] (ms) at (6,-1) {};
	
	\node[circle, scale=0] (rns) at (14,1) {};
	\node[circle, scale=0] (rs) at (14,-1) {};
	
	\node[circle] (dns) at (6.5,1) {$D_{ns}$};
	\node[circle] (ds) at (6.5,-1) {$D_s$};

	\draw[->, anchor=north] (c.east) -- node[below] {$\alpha$} node[above, scale=0.8] {$\casfunc_{p_N}(0,1)$} (ct.west);
	
	\path[->] (ct) edge[sloped, above, scale=0.8] node[align=center] {\textcolor{blue}{true}\\$\casfunc_q(0,1)$} (cns);
	\path[->] (ct) edge[sloped, below, scale=0.8] node[align=center] {$s \circ \casfunc_q(0,1)$\\\textcolor{blue}{false}} (cs);
	\draw[->, anchor=south, sloped, scale=0.7] (cns) -- node {}(tns.west);
	\draw[->, anchor=south, sloped, scale=0.7] (cns) -- node {}(bns.west);
	\draw[->, sloped, dashed, above, scale=0.8] (cns) -- node[] {$\beta _{ns}$}  (mns);
	\draw[->, anchor=south, sloped, scale=0.7] (cs) -- node {}(ts.west);
	\draw[->, anchor=south, sloped, scale=0.7] (cs) -- node {}(bs.west);
	\draw[->, sloped, dashed, above, scale=0.8] (cs) -- node[] {$\beta _{s}$}  (ms);
	
	\path[->] (ds) edge[sloped] node[below, align=center] {$\gamma$} node[above, scale=0.8, align=center] {\textcolor{blue}{$\qquad \qquad \qquad $ false / fail}\\$CRASH \circ \casrecover_{p_N}(0,1)$} (rs.west);
	
	\path[->] (dns) edge[sloped] node[below, align=center] {$\gamma$} node[above, scale=0.8, align=center] {\textcolor{blue}{$\qquad \qquad \qquad$ false / fail}\\$CRASH \circ \casrecover_{p_N}(0,1)$} (rns.west);
	
	\end{scope}
	\end{tikzpicture}
	\caption{The induction step of Theorem \ref{theorem:cas-lower-bound}. Response of completed operations are in blue font.}
	\label{fig:impos-cas}
\end{figure*}

\begin{theorem}
	\label{theorem:cas-lower-bound}
Let $A$ be an $N$-process obstruction-free implementation of a CAS object $O$ (using any primitive operations) satisfying durable linearizability and detectability, assuming values from a domain $\mathcal{V}$ of size at least $N$. Then $A$ has at least $2^{N-1}$ different reachable configurations, no two of which are memory-equivalent.
\end{theorem}

\begin{proof}
We prove the claim by induction on $N$. The claim holds trivially for $N=1$. Assume the claim holds for $N$ and let $A$ be an ($N$+1)-process  obstruction-free recoverable implementation of a CAS object $O$ satisfying durable linearizability and detectability, assuming values from a domain $\mathcal{V}$ of size at least $N+1$. Denote the processes using $A$ by $p_1, p_2, \ldots, p_N, q$. Assume also WLOG that $\{0, 1, \ldots, N\} \subseteq {\mathcal{V}}$ and that $O$'s value in the initial configuration $C$ is $0$. Starting from $C$, we let $p_N$ perform a $\casfunc_{p_N}(0,1)$ operation until it is about to perform the first modifying step $s$ with respect to a $\casfunc_q(0,1)$ operation by process $q$. We denote by $\alpha$ the prefix of $p_N$'s execution up to (but not including) step $s$. That is, a solo execution of $\casfunc_q(0,1)$ by $q$ after $C \circ \alpha$ returns \emph{true}, while its solo execution after $C \circ \alpha \circ s$ returns \emph{false}. Such a modifying step exists because $A$ is obstruction-free, a solo execution of $\casfunc_q(0,1)$ starting from $C$ returns \emph{true}, but a solo execution of $\casfunc_q(0,1)$ after $\casfunc_{p_N}(0,1)$ completes returns \emph{false}.
We define the following two configurations:
	$$ C_{ns} = C \circ \alpha \circ \casfunc_q(0,1) \qquad C_s = C \circ \alpha \circ s \circ \casfunc_q(0,1) $$
\noindent that are reached after $\casfunc_q(0,1)$ returns. The executions leading to $C_{ns}$ and $C_s$ exist, since $A$ is obstruction-free. From linearizability, $O$'s value is $1$ in both $C_s$ and $C_{ns}$, that is, a \readfunc(\ ) (resp. a \casfunc$(1, 0)$) operation on $O$, performed to completion immediately after either $C_s$ or $C_{ns}$, must return $1$ (resp. return \emph{true} and change the object's value to $0$). Moreover, $\casfunc_q(0,1)$ returns \emph{true} in the execution leading to $C_{ns}$  and \emph{false} in the execution leading to $C_{s}$.
Notice that all processes but $p_N$ are idle in $C_{ns}$. Moreover, if we fix $C_{ns}$ as an initial configuration, halt $p_N$ starting from this point on and restrict processes $p_1, \ldots, p_{N-1}, q$ to perform CAS operations with arguments from the domain ${\mathcal{V}} \setminus \{0\}$, then we obtain an $N$-process algorithm $A'$ for which we can apply the induction hypothesis. Thus, there is a set $S_1$ of at least $2^{N-1}$ configurations reachable from $C_{ns}$ in $p_N$-free executions, no two of which are memory-equivalent. The same argument can be applied to configuration $C_s$, resulting in a second set $S_2$ of $2^{N-1}$ configurations reachable from $C_s$, no two of which are memory-equivalent.

To complete the proof, we show that no configuration reachable from $S_1$ can be memory-equivalent with any configuration reachable from $S_2$. Intuitively, this is because the modifying step $s$ must write to shared memory and does not modify any of $p_N$'s local variables. Thus, if we can extend both $C_s$ and $C_{ns}$ and reach two memory-equivalent configurations, then $p_N$ will behave the same after a crash in both of them, which leads to a contradiction. The formal proof follows.
	
Assume towards a contradiction that there exist executions $\beta_{ns}$ and $\beta_s$ starting from $C_{ns}$ and $C_s$, respectively, leading to memory-equivalent configurations $D_{ns} \in S_1$ and $D_s \in S_2$. Note that both $\beta_{ns}$ and $\beta_s$ are $p_N$-free. We extend $D_{ns}$ by a system crash, followed by a recovery of $p_N$, followed by an execution in which $p_N$ performs its recovery function $\emph{\text{\casrecover}}_{p_N}(0,1)$ to completion. Denote this extension by $\gamma$. Since $\emph{\text{\casfunc}}_{p_N}(0,1)$ can be linearized only after $\emph{\text{\casfunc}}_q(0,1)$ (as the latter returns \emph{true} in the execution leading to $C_{ns}$) and as $O$'s state is positive all throughout $\beta_{ns}$, $\emph{\text{\casrecover}}_{p_N}(0,1)$ must return either \emph{false} or \emph{fail} in $\gamma$.

Since all of $p_N$'s local variables are reset after the crash,
and since $D_{ns}$ and $D_s$ are memory-equivalent, if the system crashes after $D_s$, and then $p_N$ recovers by performing $\emph{\text{\casrecover}}_{p_N}(0,1)$ to completion, it will return \emph{false} or \emph{fail} as well. However, as $\emph{\text{\casfunc}}_q(0,1)$ returned \emph{false} in the execution leading to $C_s$, from linearizability, $\emph{\text{\casrecover}}_{p_N}(0,1)$ must return \emph{true}. This is a contradiction.
\end{proof}


%% file: space-lower-bound.tex
\section{Detectable Algorithms Require Auxiliary State}
\label{section: auxiliary-state}

The detectable algorithms we presented in Sections
\ref{section: read-write}-\ref{section: CAS} use \emph{auxiliary state}, provided via the non-volatile $Ann_p.CP$ field used for managing checkpoints.
As we've mentioned in Section \ref{section:introduction}, detectable algorithms often require auxiliary state that helps them infer where in the execution the failure occurred \cite{AttiyaBH18,FriedmanHMP18}. We now formalize this notion.

\begin{definition}
\label{definition:aux}
Let $Op$ be a recoverable operation. We say that auxiliary state is provided to $Op$ via NVM, if in-between every two successive invocations of $Op$, a write is made to a non-volatile variable that can be accessed by $Op$. We say that auxiliary state is provided to $Op$ via operation arguments, if the arguments to $Op$ contain data not specified by the object's abstract operation.
\end{definition}

For example, in our model, the system provides auxiliary state via $Ann_p.CP$, since it sets its value to $0$ before any invocation of a recoverable operation by $p$. The queue algorithm of \cite{FriedmanHMP18} provides auxiliary state via unique operation identifiers passed as arguments.

In the following, we prove that \emph{the usage of auxiliary state is inevitable for obtaining detectable implementations for a large class of objects} that includes read/write, CAS, resettable test-and-set and FIFO queue objects.
\remove{
Some works considered a weaker notion of detectability than ours, where upon a recovery a process can infer the response of its last \emph{completed} operation. For example, \cite{DBLP:conf/opodis/BerryhillGT15, DBLP:conf/spaa/CohenGZ18} presented universal constructions that can be made detectable under this weaker definition. The proofs we present below do not hold for this definition. In contrast, in this work we consider a variant of detectability which requires a process to infer to response of its last invoked operation, hence the recovery function needs to distinguish between different invocations of the same operation, as we later formalize.
Our proof holds for any implementation satisfying the following weak progress condition (which extends \cite[Definition 2]{attiya2008tight}).
}

\begin{definition}
A recoverable implementation $A$ satisfies \emph{weak obstruction-freedom} if for every process $p$, starting from any configuration (reached by an execution of $A$) in which all processes other than $p$ are idle, if $p$ runs by itself while performing an operation or a recovery function and incurs no crashes, then it completes it.
\end{definition}

It is easily seen that weak obstruction-freedom is satisfied by any recoverable algorithm that satisfies obstruction-freedom or deadlock-freedom. Our proof holds for the class of \emph{doubly-perturbing} objects, a notion we define next.


\begin{definition}
\label{def:doubly-perturbing}
Given an object $O$ and a sequential history $H$, we say that an operation $\Op$ is \emph{perturbing after H}, if there exists an operation $\Op'$ by a different process such that $\Op'$ returns different responses in $H \circ \Op \circ \Op'$ and in $H \circ \Op'$. We also say that \emph{Op is perturbing with respect to $\Op'$ after H}. We say that $O$ is \emph{doubly-perturbing} if there exists an operation $\Op_p$ by some process $p$ and a sequential history $H_1$ of $O$ such that the following conditions hold:
\begin{enumerate}
\item  $\Op_p$ is perturbing with respect to some $\Op'$ after $H_1$.
\item $H_1 \circ \Op_p \circ \Op'$ has a $p$-free extension, resulting in a sequential history $H_2$, such that $\Op_p$ is perturbing after $H_2$.
\end{enumerate}

We say that $\Op_p$ \emph{witnesses} that $O$ is doubly-perturbing.

\end{definition}

It is easily shown that many widely-used objects
are doubly-perturbing. We now prove that a read/write object is doubly perturbing. In the appendix, we provide proofs that establish that the counter, CAS, fetch-and-add and FIFO queue objects are doubly-perturbing as well.

\begin{lemma}
	A read/write object is doubly-perturbing.
\end{lemma}

\begin{proof}
	Consider a read/write object $O$ over a domain of values including at least two distinct values $v_0$, $v_1$, initialized to $v_0$. We claim that $write_p(v_1)$ witnesses that $O$ is doubly-perturbing. For any process $q \neq p$, $write_p(v_1)$ is perturbing w.r.t. $read_q$ after the empty sequential history. This satisfies the first condition of Definition \ref{def:doubly-perturbing}. Moreover, $write_p(v_1) \circ read_q$ can be extended by a $write_q(v_0)$ operation, resulting in a sequential history $H_2=write_p(v_1) \circ read_q \circ write_q(v_0)$, such that (a second instance of) $write_p(v_1)$ is perturbing w.r.t. (a second instance of) $read_q$ operation after $H_2$. This satisfies the second condition of Definition \ref{def:doubly-perturbing}. Thus, $O$ is doubly-perturbing.
\end{proof}

The notion of doubly-perturbing objects bears similarity to the notion of perturbable objects defined by Jayanti et al. \cite{JayantiTT00}. Although most common perturbable objects (including read/write, counter, compare-and-swap, swap and fetch-and-add objects) are also doubly-perturbing, the two classes of objects are incomparable.
We now prove that there are perturbable objects that are not doubly-perturbing (e.g., a max register \cite{AspnesAC12}). In the appendix we prove that a bounded counter is doubly-perturbing but not perturbable.

\begin{lemma}
	\label{lemma:max-register}
	A max register object is not doubly-perturbing.
\end{lemma}

\begin{proof}
	Consider a max register object $O$, which supports a $writeMax(v)$ operation and a $read()$ operation that returns the largest value written to the register preceding it. Clearly, a $read$ operation cannot witness that $O$ is doubly-perturbing, as it cannot be observed by other operations. As for a $writeMax(v)$ operation by some process $p$, once the operation is lineraized any read must return a value $v$ or higher. Thus, invoking $writeMax(v)$ for a second time cannot modify $O$'s value, and thus cannot change the response of any other operation.
\end{proof}

\newcommand{\maxfunc}{\mbox{\sc Write-Max}}
\newcommand{\maxrecover}{\mbox{\sc Write-Max.Recover}}

Lemma \ref{lemma:max-register} proves that max register is not doubly-perturbing. However, it is known to be perturbable \cite{JayantiTT00}. Thus, Theorem \ref{theorem:auxiliary-data-required} below does not apply for a max register object. This raises the question of whether there is an obstruction-free implementation of a max-register satisfying durable linearizability and detectability that does not use auxiliary state. We now show that this is indeed the case. 

Algorithm \ref{max-register-algorithm} presents the pseudo-code of a detectable max register implementation that uses no auxiliary state.
The max register is composed of an integers array $MR$, such that process $p$ is associated with entry $MR[p]$. To perform \maxfunc $(val)$, $p$ simply checks if the value in $MR[p]$ is smaller than $val$, and if so updates it. To perform \readfunc, $p$ uses a local array. For simplicity of presentation, we use an array-assignment/comparison notation as a shorthand for copying/comparing the array entry by entry. Process $p$ repeatedly copies the contents of $MR$ until the first successful double collect. Upon success, $p$ managed to obtain a valid snapshot of $MD$, so the largest value in $MD$ was the value of the max register at some point in the execution interval of \readfunc.
The recovery function of each of these operations simply re-invokes the operation, and thus the code is not shown.

\begin{figure*}[!t]
	\removelatexerror
	
	\begin{algorithm}[H]
		
		\footnotesize
		
		\begin{flushleft}
			\textbf{Non-Volatile Shared variables}: integer array $MR[N]$ initially all 0 \\
		\end{flushleft}
		
		\begin{multicols*}{2}
			
			\begin{procedure}[H]
				\caption{() \small \maxfunc\ (val)}
				
				\uIf {$MR[p] < val$} {
					$MR[p] := val$	
				}
				\KwRet \Ack
			\end{procedure}
			
			\columnbreak
			
			\begin{procedure}[H]
				\caption{() \small \readfunc\ ()}
				
				$a[N]$, initially all 0 \;
				\While{$a \neq MR$}{
					$a := MR$
				}
				$res :=$ highest value in $a$ \;
				$Ann_p.result := res$ \;
				\KwRet $res$
			\end{procedure}
			
		\end{multicols*}
		\caption{recoverable Max-Register object $O$, code for process $p$.}
		\label{max-register-algorithm}
	\end{algorithm}
	
\end{figure*}

Next, we present our impossibility result.
The following two definitions are required for the proof.  We say that configurations $C,D$ are \emph{indistinguishable} to process $p$, denoted by $C \stackrel{p}{\sim} D$,  if the values of all shared-memory variables, as well as those of $p$'s local variables, are the same in both configurations. We say that $C \stackrel{Q}{\sim} D$ for a set of processes $Q$ if $C \stackrel{p}{\sim} D$ for any $p \in Q$.
Figure \ref{fig:impos} depicts the structure of the execution constructed by our proof.
We note that Theorem \ref{theorem:auxiliary-data-required} holds regardless of the primitive operations used by the implementation.

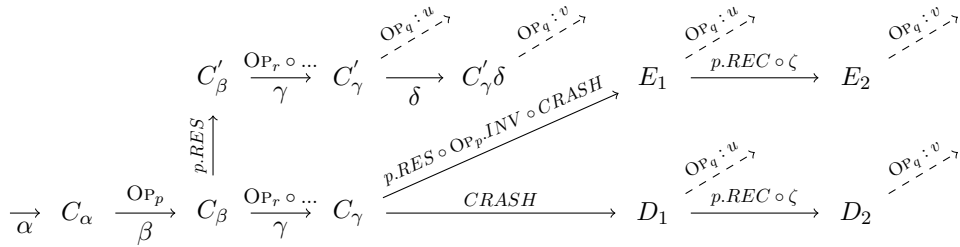
\begin{figure*}[b]
	\centering
	\begin{tikzpicture}
	\begin{scope}[scale=0.9, transform shape]
	\node[circle, minimum height=1.1cm, dotted] (ca) at (0,0) {$C_{\alpha}$};
	\node[circle, scale=0, draw] (c0) at (-1,0) {};
	\node[circle, minimum height=1.1cm, dotted] (cb) at (2,0) {$C_{\beta}$};
	\node[circle, minimum height=1.1cm, dotted] (cg) at (4,0) {$C_{\gamma}$};
	\node[circle, minimum height=1.1cm, dotted] (d1) at (8.5,0) {$D_1$};
	\node[circle, minimum height=1.1cm, dotted] (d2) at (11.5,0) {$D_2$};

	\node[circle, minimum height=1.1cm, dotted] (cbt) at (2,2) {$C_{\beta}^{'}$};
	\node[circle, minimum height=1.1cm, dotted] (cgt) at (4,2) {$C_{\gamma}^{'}$};
	\node[circle, minimum height=1.1cm, dotted] (cgtd) at (6,2) {$C_{\gamma}^{'} \delta$};
	
	\node[circle, minimum height=1.1cm, dotted] (e1) at (8.5,2) {$E_1$};
	\node[circle, minimum height=1.1cm, dotted] (e2) at (11.5,2) {$E_2$};
	
	\node[circle, scale=0] (cgtr) at (5.5,2.9) {};
	\node[circle, scale=0] (cgtdr) at (7.5,2.9) {};
	\node[circle, scale=0, draw] (e1r) at (10,2.9) {};
	\node[circle, scale=0, draw] (d1r) at (10,0.9) {};
	\node[circle, scale=0, draw] (e2r) at (13,2.9) {};
	\node[circle, scale=0, draw] (d2r) at (13,0.9) {};

	
	\draw[->, anchor=north] (c0.east) -- node {$\alpha$} (ca.west);
	
	\draw[->] (ca) -- node[above, scale=0.8] {$\Op_p$} node[below] {$\beta$} (cb);
	\draw[->] (cb) -- node[sloped, above, scale=0.7] {$p.RES$}  (cbt);
	
	\draw[->] (cb) -- node[above, scale=0.8] {$\Op_r\circ ...$} node[below] {$\gamma$} (cg);
	\draw[->] (cbt) -- node[above, scale=0.8] {$\Op_r\circ ...$} node[below] {$\gamma$} (cgt);
	\draw[->] (cgt) -- node[below] {$\delta$} (cgtd);
	\draw[->] (cg) -- node[sloped, above, scale=0.7] {$p.RES\circ \Op_p.INV \circ CRASH$}  (e1);
	
	\draw (cg)  edge[->, anchor=south, scale=0.8] node {\small $CRASH$} (d1);
	\draw (d1)  edge[->, anchor=south, scale=0.8] node {\small $p.REC \circ \zeta$} (d2);
	\draw (e1)  edge[->, anchor=south, scale=0.8] node {\small $p.REC \circ \zeta$} (e2);
	
	\draw[->, anchor=south, sloped, dashed, scale=0.7] (cgt) -- node {$\Op_q:u$}(cgtr.west);
	\draw[->, anchor=south, sloped, dashed, scale=0.7] (cgtd) -- node {$\Op_q:v$}(cgtdr.west);
	\draw[->, anchor=south, sloped, dashed, scale=0.7] (e1) -- node {$\Op_q:u$}(e1r.west);
	\draw[->, anchor=south, sloped, dashed, scale=0.7] (d1) -- node {$\Op_q:u$}(d1r.west);
	\draw[->, anchor=south, sloped, dashed, scale=0.7] (e2) -- node {$\Op_q:v$}(e2r.west);
	\draw[->, anchor=south, sloped, dashed, scale=0.7] (d2) -- node {$\Op_q:v$}(d2r.west);
	\end{scope}
	\end{tikzpicture}
	\caption{The execution of $A$ constructed in the proof of Theorem \ref{theorem:auxiliary-data-required}.}
	\label{fig:impos}
\end{figure*}



\begin{theorem}
	\label{theorem:auxiliary-data-required}
Let $A$ be a weak obstruction-free recoverable implementation of a doubly-perturbing object $O$ satisfying durable linearizability and detectability and let $Op$ be a recoverable operation of $Op$. Then either auxiliary state is provided to $Op$ via NVM, or it is provided to $Op$ via operation arguments.
\end{theorem}

\begin{proof}
Assume towards a contradiction that auxiliary state is not provided to $Op$ neither via NVM nor via the arguments to $Op$, so the arguments to operations in $A$ are identical to those applied to the implemented object. For example, if $O$ is a read/write object, our assumption means that read operations receive no arguments and a write operation invoked for writing value $v$ receives a single argument whose value is $v$. We construct an execution of $A$ that will establish that $A$ violates durable linearizability, thus reaching a contradiction.

\remove{
\oh{
We first provide an intuition for the construction. Let $\Op_p$ be an operation by some process $p$ witnessing that $O$ is doubly-perturbing. That is, there exists an history in which $\Op_p$ is executed twice by process $p$, and each time it perturbs some other operation. A key point in the construction is that an invocation as well as a response step of $p$ does not effect any other process, as those steps are local and does not write to shared memory. Moreover, under our assumption, if a crash occurs immediately before $p$ returns from the first occurrence of $\Op_p$, or if we let $p$ complete $\Op_p$, invoke its next $\Op_p$, and only then a crash occurs, upon recovery $p$ can not distinguish between the two resulting configurations. This is due to the fact that no non-volatile variable is written between the two steps of $p$, and the arguments passed to both operations are identical. Hence, upon recovery all shared-memory variables, as well as $p$'s local variables that have been initialized by the crash, are the same. Therefore, the recovery function of $\Op_p$ in both cases is identical. However, the first $\Op_p$ was already linearized, and thus the recovery function in the first case must return its response, while the second $\Op_p$ was yet to take any steps, and the recovery function in the second case should return fail. A formal proof follows.
}} 
	
Since $O$ is doubly-perturbing, there exists an operation $\Op_p$ by some process $p$, witnessing that $O$ is doubly-perturbing. Thus, from the first condition of Definition \ref{def:doubly-perturbing},
there is a sequential history $H_1$ such that $\Op_p$ is perturbing w.r.t. to some operation $\Op_r$ after $H_1$, for $r \neq p$. Since $A$ satisfies weak obstruction-freedom, there exists an execution $\alpha$ in which processes perform their operations according to $H_1$ in a sequential manner, starting from an initial configuration, leading to a configuration $C_\alpha$. Starting from $C_\alpha$, let $p$ apply its $\Op_p$ operation and halt just before returning; denote the resulting configuration by $C_\beta$. Starting from $C_\beta$, let $p$ complete its $\Op_p$ operation on $O$ by returning from it and denote the resulting configuration by $C'_\beta$. As $A$ satisfies weak obstruction-freedom, the executions leading to $C_\beta$ and $C'_\beta$ exist, since in both $p$ runs solo starting from a quiescent configuration. Since returning a response does not change any shared variable and no auxiliary state is provided via NVM, we get: $\text{(1) }\forall q \neq p: C_\beta \stackrel{q}{\sim} C'_\beta.$
Notice that $C'_\beta$ is a quiescent configuration.
Since $\Op_p$ witnesses that $O$ is  doubly-perturbing, from the second condition of Definition \ref{def:doubly-perturbing}, the history $H_1 \circ \Op_p \circ \Op_r$ has a $p$-free extension resulting in a sequential history $H_2$ (extending $H_1$) such that $\Op_p$ is perturbing after $H_2$. Assume it is perturbing w.r.t. an operation $\Op_q$ (for some $q \neq p$). Since $A$ is weak obstruction-free, there exists an execution $\gamma$ starting from $C'_\beta$ in which first $r$ performs to completion $\Op_r$, followed by a $p$-free sequential execution corresponding to the extension of $H_1 \circ \Op_p \circ \Op_r$, resulting in configuration $C'_\gamma$. Moreover, let $\delta$ denote the solo execution of $\Op_p$ after $C'_\gamma$, then $\Op_q$ returns different responses when executed after $C'_\gamma$ and after $C'_\gamma \delta$.
Since $\gamma$ is a $p$-free execution, and from Equation (1), 
we get that $\gamma$ is also an execution starting from $C_\beta$ such that: $\text{(2) }\forall q \neq p: C_\gamma \stackrel{q}{\sim} C'_\gamma.$	
$\Op_r$ was performed to completion in $\gamma$. Since $\Op_r$ returns the same response in the executions leading to $C'_\gamma$ and to $C_\gamma$ and as $\Op_p$ perturbed its response, it follows that the linearization point of $\Op_p$ precedes that of $\Op_r$ in both executions.
	
Next, consider the following two configurations: $D_1 = C_\gamma \circ CRASH $ and $E_1 = C_\gamma \circ (\Op_p.RES) \circ  (\Op_p.INV) \circ CRASH $. Thus, in the execution leading to $D_1$, the system crashes just before $p$ returns from its first $\Op_p$ operation, whereas, in the execution leading to $E_1$, $p$ completes that operation (thus $p$ returns its response), invokes a second $\Op_p$ operation and then the system crashes immediately after $p$'s invocation. Note that both these executions exist, since all processes but $p$ are idle in $C_\gamma$ and $A$ satisfies weak obstruction-freedom. Since neither responses nor invocations change any shared variables and since no auxiliary state is provided, neither via NVM nor via operation arguments, and as a crash reset all local variables of $p$ we get: $\text{(3) }D_1 \stackrel{P}{\sim} E_1$.
Next, starting from $E_1$, we extend the execution by recovering $p$, letting it perform a solo execution of its recovery function to completion, denoted by $\zeta$. We let $E_2 = E_1 \circ (p.REC) \circ \zeta$ denote the resulting configuration. $\zeta$ exists, since all processes but $p$ are idle in $E_1$ and $A$ satisfies weak obstruction-freedom. From Equation (3), $p$ performs the same execution $\zeta$ after $E_1 \circ (p.REC)$ and after $D_1 \circ (p.REC)$, hence $D_2 = D_1 \circ (p.REC) \circ \zeta$ exists and we have: $\text{(4) }D_2 \stackrel{P}{\sim} E_2$.

	
	Finally, starting from configuration $E_2$, extend the execution by letting process $q \neq p$ perform $\Op_q$ to completion. $\Op_q$ completes, since $E_2$ is quiescent and $A$ satisfies weak obstruction-freedom. Assume that $\Op_q$ returns value $v$. From Equation (4), $\Op_q$ returns $v$ also if it is performed starting from configuration $D_2$.
	Notice that $p$ cannot return \fail\ in its recovery function in $\zeta$. This is due to the fact that in the execution leading to $D_1$ $p$ executed a single $\Op_p$ operation which perturbs the response of $\Op_r$. Hence, $\Op_p$ must be linearized before $\Op_r$, and the recovery function returns its response. In particular, the recovery function of $p$ executed in $E_2$ returns a response different then \fail, and thus the second $\Op_p$ must have a linearization point in the execution leading to $E_2$.
	Moreover, since $\Op_p$ is perturbing w.r.t. $\Op_q$ after $H_2$, executing $\Op_q$ starting from $E_1$ returns a response $u \neq v$.
	That is, the execution of $\Op_p$ in $\zeta_p$ perturbs the response of $\Op_q$. From Equation (3), we get that executing $\Op_q$ starting from $D_1$ returns $u$ as well.
	It follows that $p$'s recovery function, executed in $\zeta$ after $D_1$, perturbs the response of $\Op_q$ performed after $D_2$. However, $p$ executed \textit{a single} $\Op_p$ operation in the execution leading to $D_2$, which was linearized before $\Op_r$, while a second linearization point of $\Op_p$ must exist in the execution of $\Op_p$'s recovery function ($\zeta$), since it perturbs the response of $\Op_q$. Thus, $\Op_p$ has two linearization points in the same execution. This is a contradiction to our assumption that $A$ satisfies durable linearizability.
	
\end{proof}


%% file: discussion.tex
\section{Discussion}
\label{section: discussion}


Several works propose different correctness conditions for the crash-recovery model \cite{Aguilera2003StrictLA, BerryhillGT15, GuerraouiL04}. Izraelevitz et al. \cite{IzraelevitzMS16} presented \emph{durable linearizability} (DL) that assumes a system-wide crash model. Roughly speaking, it requires that after the system recovers from a crash, the state of each object reflects a consistent sub-history that contains all the operations that completed before the crash. Detectability \cite{FriedmanHMP18} requires, in addition, that a process recovering from a crash that occurred while it was performing an operation $Op$ can infer if $Op$ took effect (was linearized) and, in such case, obtain its response. This allows for a ``client'' operation that called $Op$ to continue its execution after the crash, since $Op$'s response is made available to it. If $Op$ was not linearized, the client operation can choose whether or not to re-invoke it.

Attiya et al. \cite{AttiyaBH18} formalized a strict variant of this notion as \emph{nesting-safe recoverable linearizability} (NRL). NRL requires that $Op.\texttt{Recover}$ complete the crashed operation and persist its response before returning, thus allowing the client operation access to this response. NRL is stricter than detectability \cite{FriedmanHMP18}, since the latter provides the client code with the flexibility of choosing whether or not to re-invoke a crashed operation that was not linearized (upon receiving a \fail\ response), whereas NRL will re-attempt such an operation again and again until it completes (if it ever does). 

We note that NRL implies DL and detectability, as each crashed operation must complete, persist the response, and is linearized before $Op.\texttt{Recover}$ completes. As for the other direction, given an implementation that satisfies both DL and detectability, one can easily transform it into an implementation satisfying NRL by having the recovery function $Op.\texttt{Recover}$ invoke $Op$ again instead of returning a \fail\ response. Therefore, all algorithms presented in this paper can be easily be transformed to satisfy NRL.

The \emph{private cache model} \cite{FriedmanHMP18,IzraelevitzMS16} is an abstract model, where primitive operations are applied directly to NVM and processes do not share a cache. In the more realistic \emph{shared cache model} \cite{IzraelevitzMS16}, the main memory is non-volatile, while primitive operations are applied to a volatile shared cache. In this model, explicit persistency instructions are used in order to ensure that writes get persisted to the NVM and do so in the right order. Upon a crash, the shared cache content, as well as any the values of local variables, are lost. 
It is easily seen that our impossibility and lower-bound results hold also in the shared cache model. As for our algorithms, for presentation simplicity, we specified and analyzed them assuming the private cache model. A simple syntactic transformation was proposed by \cite{IzraelevitzMS16}, where persistency instructions are added to the code, thus transforming it to an implementation for the shared cache model. By applying this transformation to our algorithms, they maintain their correctness (as well as their space complexity) in the shared cache model as well.

An alternative approach for achieving recoverability is the design of persistent universal constructions. For example, \cite{CohenGZ18} designed a persistent log-based universal construction that requires only one round trip to NVRAM per operation, which is optimal. However, logging imposes significant overheads in time and space, which are even more pronounced for concurrent
data structures, where there is an extra cost of synchronizing
accesses to the log.
\cite{BerryhillGT15} presented a recoverable variant of Herlihy's wait-free universal construction.
These constructions do not provide detectability as defined by \cite{FriedmanHMP18}. However, upon recovery processes are able to determine which operations were linearized before the crash and obtain their response.
Such implementations does not require auxiliary state, as upon recovery process can not infer whether its last invoked operation was linearized. For example, assume process $p$ performs the same operation $Op$ twice, and a crash occurs during its second instance of $Op$. Upon recovery, $p$ is able to conclude that its last linearized operation was $Op$, and obtain its response. However, $p$ can not tell what instance of $Op$ was linearized. Unlike it, detectability requires $p$ to infer if the second instance of $Op$, during which the crash occurred, was linearized. 

Ben-David et al. \cite{Ben-DavidBFW19} showed that any implementation using only read, write and \CAS\ primitives can be made
detectable by partitioning the code into \emph{capsules},
each containing a single \CAS\ followed by several reads,
and replacing each \CAS\ with its recoverable version.

Some recent works study the \emph{recoverable mutual exclusion} (RME) problem, defined by Golab and Ramaraju \cite{GolabH17recoverable,GolabH18recoverable,GolabR16recoverable,JayantiJJ18,JayantiJ17}.
Much work was also done on implementing persistent transactional memory frameworks (see \cite{CorreiaFR18, KolliPSCW16, LiuZCQWZR17, VolosTS11}).
Several previous works investigated recoverable implementations for specific data-structures. Friedman et al. \cite{FriedmanHMP18} presented concurrent lock-free queue algorithms exhibiting different
tradeoffs between consistency and efficiency.
Coburn et al. \cite{CoburnCAGGJS11} presented \emph{NV-heaps}. Several works proposed persistent algorithms for index trees
\cite{YangWCWYH15,LeeLSNN17, OukidLNWL16,YangWCWYH15}.

In this work, we presented the first bounded-space detectable CAS and read/write algorithms. Detectable algorithms have the advantage of supporting composability.  On the downside, as we established, this comes with a price tag in terms of space complexity and the need to provide recoverable operations and recovery functions with auxiliary data.

\paragraph*{Open Problems}
Theorem \ref{theorem:cas-lower-bound} establishes a space lower bound of $\Omega(N)$ shared bits on any obstruction-free detectable implementation of a \CAS\ object, thus establishing that Algorithm \ref{alg:cas} is asymptotically space optimal. However, Algorithm \ref{alg:cas} uses \emph{a single} $\Omega(N)$-bits shared variable. Finding such an algorithm which uses registers of size $O(\log N)$ bits, or alternatively proving that such an algorithm is impossible, is an interesting research direction.

No (non-trivial) space lower bound for a detectable read/write object is known and finding the tight bound is another open question. Finally, exploring the tradeoff between space and time complexity for detectable implementations, as well as the tradeoff between the complexities of a recoverable operation and its recovery function, is another interesting avenue for future work.

%% file: appendix-doubly-perturbing.tex
\section{Doubly-Perturbing Objects}
\label{section: appendix-doubly-perturbing}

\begin{lemma}
	\label{lemma:counter}
	A counter object is doubly-perturbing.
\end{lemma}

\begin{proof}
	Consider a counter object $O$ over a domain of values including at least three distinct values $v_0$, $v_0+1$, $v_0+2$, initialized to $v_0$. We claim that $Increment_p$ witnesses that $O$ is doubly-perturbing. For any process $q \neq p$, $Increment_p$ is perturbing w.r.t. $read_q$ after the empty sequential history. This satisfies the first condition of Definition \ref{def:doubly-perturbing}. Moreover, $Increment_p \circ read_q$ can be extended by a an empty $p-free$ extension, resulting in a sequential history $H_2=Increment_p \circ read_q$, such that (a second instance of) $Increment_p$ is perturbing w.r.t. (a second instance of) $read_q$ operation after $H_2$. This satisfies the second condition of Definition \ref{def:doubly-perturbing}. Thus, $O$ is doubly-perturbing.
\end{proof}

Lemma \ref{lemma:counter} proves that a bounded counter, supporting only $\{0,1,2\}$ values, is doubly-perturbing. However, such a bounded counter is not perturbable, since clearly an operation $Op$ can change its response at most twice.

\begin{lemma}
	A compare-and-swap object is doubly-perturbing.
\end{lemma}

\begin{proof}
	Consider a compare-and-swap object $O$ over a domain of values including at least two distinct values $v_0$, $v_1$, initialized to $v_0$. We claim that $CAS_p(v_0,v_1)$ witnesses that $O$ is doubly-perturbing. For any process $q \neq p$, $CAS_p(v_0,v_1)$ is perturbing w.r.t. $CAS_q(v_0,v_1)$ after the empty sequential history. This satisfies the first condition of Definition \ref{def:doubly-perturbing}. Moreover, $CAS_p(v_0,v_1) \circ CAS_q(v_0,v_1)$ can be extended by a $CAS_q(v_1,v_0)$ operation, resulting in a sequential history $H_2=CAS_p(v_0,v_1) \circ CAS_q(v_0,v_1) \circ CAS_q(v_1,v_0)$, such that (a second instance of) $CAS_p(v_0,v_1)$ is perturbing w.r.t. (a second instance of) $CAS_q(v_0,v_1)$ operation after $H_2$. This satisfies the second condition of Definition \ref{def:doubly-perturbing}. Thus, $O$ is doubly-perturbing.
\end{proof}

\begin{lemma}
	A fetch-and-add object is doubly-perturbing.
\end{lemma}

\begin{proof}
	Since fetch-and-add object support an addition of the value $1$, the proof from Lemma
	\ref{lemma:counter} will hold here too. More formally, consider a fetch-and-add object $O$ over a domain of values including at least three distinct values $v_0$, $v_0+1$, $v_0+2$, initialized to $v_0$. We claim that $FAA_p(1)$ witnesses that $O$ is doubly-perturbing. For any process $q \neq p$, $FAA_p(1)$ is perturbing w.r.t. $read_q$ after the empty sequential history. This satisfies the first condition of Definition \ref{def:doubly-perturbing}. Moreover, $FAA_p(1) \circ read_q$ can be extended with the empty $p-free$ extension, resulting in a sequential history $H_2=FAA_p(1) \circ read_q$, such that (a second instance of) $FAA_p(1)$ is perturbing w.r.t. (a second instance of) $read_q$ operation after $H_2$. This satisfies the second condition of Definition \ref{def:doubly-perturbing}. Thus, $O$ is doubly-perturbing.
\end{proof}
\begin{lemma}
	A FIFO queue object is doubly-perturbing.
\end{lemma}

\begin{proof}
	Consider a FIFO queue object $O$ over a domain of values including at least two distinct values $v_0$, $v_1$, initialized to an empty queue. We claim that $Deq_p$ witnesses that $O$ is doubly-perturbing. After the empty sequential history, we perform the following sequence of operations $Enq_p(v_0) \circ Enq_p(v_1)$, and denote the resulting history as $H_1$.
	For any process $q \neq p$, $Deq_p$ is perturbing w.r.t. $Deq_q$ after $H_1$. This satisfies the first condition of Definition \ref{def:doubly-perturbing}. Moreover, $Deq_p \circ Deq_q$ can be extended by the sequence of operations:  $Enq_q(v_0) \circ Enq_q(v_1)$ as before, resulting in a sequential history $H_2=Enq_p(v_0) \circ Enq_p(v_1) \circ Deq_p \circ Deq_q \circ Enq_q(v_0) \circ Enq_q(v_1)$, such that (a second instance of) $Deq_p$ is perturbing w.r.t. (a second instance of) $Deq_q$ operation after $H_2$. This satisfies the second condition of Definition \ref{def:doubly-perturbing}. Thus, $O$ is doubly-perturbing.
\end{proof}

\remove{
\section{System Support}
In the following we

\begin{figure*}[b]
	\removelatexerror
	
	\begin{algorithm}[H]
		
		\footnotesize
		
		\begin{flushleft}
			\textbf{Non-Volatile Shared variables}: read/write announcements array $annc[N]$, holding tuples of $\langle Op_name; params; result \rangle$, initially all tuples are set to $\bot $ \\
			\textbf{Non-Volatile Private variables}:
		\end{flushleft}
		
		\begin{multicols*}{2}

			\begin{procedure}[H]
				\caption{() \small \mbox{\sc Announce}\ (Op, params)}
				$annc[p] := \langle Op; params; \bot \rangle$ \;
				$annc[p].result := $ \small \mbox{\sc Op}(params) \;
			\end{procedure}
			
			\columnbreak
			
			\begin{procedure}[H]
				\caption{() \small \mbox{\sc Announce.Recover}\ ()}
				
				$\langle op; params; res \rangle := annc[p]$ \;
				\uIf {$res \neq \bot$} {
					\KwRet $res$
				}
				\Else{
					$annc[p].result := $ \small \mbox{\sc Op.Recover}(params) \;
					\KwRet $annc[p].result$
				}
			\end{procedure}

		\end{multicols*}
		\caption{Announce code for a compiler, and Announce Recovery code for process $p$.}
		\label{alg:system-recovery}
	\end{algorithm}
	\caption{On each call of some object operation $Op$ with the operation parameters $params$, the compiler/library should translate the call to \mbox{\sc Announce}$(Op, params)$. In case that process $p$ crashes, the system makes $p$ invoke the pre-compiled function \mbox{\sc Announce.Recover}$()$ that resides in the Non-Volatile memory, immediately upon recovery. }
\end{figure*}
}